\documentclass[letterpaper, 10 pt, conference]{ieeeconf}  
\IEEEoverridecommandlockouts                              

\title{\LARGE \bf
Stochastic Model Predictive Control with\\ Discounted Probabilistic Constraints
}

\author{Shuhao Yan, Paul Goulart and Mark Cannon
\thanks{The authors are with the Department of Engineering Science, University
of Oxford, OX1 3PJ, UK (E-mail: shuhao.yan@eng.ox.ac.uk; paul.goulart@eng.ox.ac.uk; mark.cannon@eng.ox.ac.uk)
        }
}

\usepackage{color}
\usepackage{amsmath}
\usepackage{amssymb}
\usepackage{amsfonts}
\usepackage{mathrsfs}
\usepackage{enumerate}
\usepackage{tikz} 
\usepackage[english]{babel}
\usepackage{datetime}

\newdateformat{dateUKenglish}{\THEDAY~\monthname[\THEMONTH] \THEYEAR}

\newcommand{\Ex}[2][]{\mathbb{E}_{#1}\left[#2\right]} 
\newcommand{\norm}[1]{\left\lVert#1\right\rVert}
\newtheorem{lemma}{Lemma}
\newtheorem{theorem}[lemma]{Theorem}
\newtheorem{algorithm}{Algorithm}
\newtheorem{assumption}{Assumption}

\DeclareMathOperator*{\minimise}{minimise}
\DeclareMathOperator{\tr}{tr}
\usepackage{balance}
\begin{document}

\maketitle
\thispagestyle{empty}
\pagestyle{empty}

\begin{abstract}
  This paper considers linear discrete-time systems with additive disturbances, and designs a Model Predictive Control (MPC) law to minimise a quadratic cost function subject to a chance constraint. The chance constraint is defined as a discounted sum of violation probabilities on an infinite horizon.
  By penalising violation probabilities close to the initial time and ignoring violation probabilities in the far future, this form of constraint enables the feasibility of the online optimisation to be guaranteed without an assumption of boundedness of the disturbance.
  A computationally convenient MPC optimisation problem is formulated using Chebyshev's inequality and we introduce an online constraint-tightening technique to ensure recursive feasibility based on knowledge of a suboptimal solution.
  The closed loop system is guaranteed to satisfy the chance constraint and a quadratic stability condition.

\end{abstract}

\section{INTRODUCTION}

Robust control design methods for systems with unknown disturbances must take into account the worst case disturbance bounds in order to guarantee satisfaction of hard constraints on system states and control inputs~\cite{blanchini91,kothare96,mayne05}. However, for problems with stochastic disturbances and constraints that are allowed to be violated up to a specified probability, worst-case control strategies can be unnecessarily conservative. This motivated the development of stochastic Model Predictive Control (MPC), which addresses optimal control problems for systems with chance constraints by making use of information on the distribution of model uncertainty~\cite{schwarm99,KOUVARITAKIS2010auto}.
Although capable of handling chance constraints, existing stochastic MPC algorithms that ensure constraint satisfaction in closed loop operation typically rely on knowledge of worst case disturbance bounds to obtain such guarantees~\cite{kouvaritakis2015mpcbook}.  For the algorithms proposed in~\cite{korda14,Lorenzen17ieeetr,fleming17} for example, which simultaneously ensure closed loop constraint satisfaction and recursive feasibility of the online MPC optimisation, the degree of conservativeness increases as the disturbance bounds become more conservative.

This paper ensures both closed loop satisfaction of chance constraints and recursive feasibility but does not rely on disturbance bounds, instead requiring knowledge of only the first and second moments of the disturbance input. This is achieved by formulating the chance constraint as the sum over an infinite horizon of discounted violation probabilities, and implementing the resulting constraints using Chebyshev's inequality. Control problems involving discounted costs and constraints are common in financial engineering applications (e.g.~\cite{FRANKEL2016,kouvaritakis06,kamgarpour17}), and allow system performance in the near future to be prioritised over long-term behaviour. This shift of emphasis is vital for ensuring recursive feasibility of chance-constrained control problems involving possibly unbounded disturbances.
We describe an online constraint-tightening approach that guarantees the feasibility of the MPC optimisation, and, by considering the closed loop dynamics of the tightening parameters, we show that the closed loop system satisfies the discounted chance constraint as initially specified.




The paper is organised as follows.
The control problem is described and reformulated with a finite prediction horizon in Section~\ref{section:problem
  description}.
Section~\ref{section:recursive feasibility} proposes an online constraint-tightening method for guaranteeing recursive feasibility.
Section~\ref{section: smpc algorithm} summarises the proposed MPC algorithm and derives bounds on closed loop performance.
In Section~\ref{section: epsilon sequence}, the closed loop behaviour of the tightening parameters is analysed and constraint satisfaction is proved.
Section~\ref{section:numerical example} gives a numerical example illustrating the results obtained and the paper is concluded in Section~\ref{section:conclusion}.

\textit{Notation}:
%
The Euclidean norm is denoted $\| x\|$ and we define $\norm{x}^2_{Q} := x^TQx$.
The notation $Q\succcurlyeq0$ and $R\succ0$ indicates that $Q$ and $R$ are respectively positive semidefinite and positive definite matrices, and $\tr(Q)$ denotes the trace of $Q$.
The probability of an event $A$ is denoted $\mathbb{P}(A)$. The expectation of $x$ given information available at time $k$ is denoted $\Ex[k]{x}$ and $\Ex{x}$ is equivalent to $\Ex[0]{x}$.
The sequence $\{ x_0,\ldots,x_{N-1}\}$ is denoted $\{x_i\}_{i=0}^{N-1}$. We denote the value of a variable $x$ at time $k$ as $x_k$, and the $i$-step-ahead predicted value of $x$ at time $k$ is denoted $x_{i|k}$.

\section{PROBLEM DESCRIPTION} \label{section:problem description}

Consider an uncertain linear system with model
\begin{align}
x_{k+1}=Ax_{k}+Bu_{k}+\omega_k, \label{eqn:system dynamics:original state space model} 
\end{align}
where $x_{k}\in \mathbb{R}^{n_x}$, $u_{k}\in \mathbb{R}^{n_u}$ are the system state and the control input respectively. The unknown disturbance input $\omega_k \in \mathbb{R}^{n_x}$ is independently and identically distributed with known first and second moments:
\[
  \Ex{\omega_k}=0, \quad \Ex{\omega_k\omega^{T}_{k}}=W.
\]
Unlike the approaches of \cite{KOUVARITAKIS2010auto,cannon2011ieee}, which assume the additive disturbance lies in a compact set, the disturbance $\omega_k$ is not assumed to be bounded and its distribution may have infinite support. It is assumed that the system state is measured directly and available to the controller at each sample instant.

The system \eqref{eqn:system dynamics:original state space model} is subject to the constraint
\begin{equation}
\sum_{k=0}^\infty \gamma^k \mathbb{P}\bigl( \norm{C x_k} \geq t \bigr)  \leq e, \label{eqn:constraint:probability constraint of the original form}
\end{equation}
for a given matrix $C \in \mathbb{R}^{n_c\times n_x}$, positive scalars $t, e$ and  discounting factor $\gamma \in (0,1)$. This constraint gives a special feature to the control problem that the probability of future states violating the condition $\| C x_k \| < t$ at time instants nearer to the initial time are weighted more heavily than those in the far future. For simplicity we refer to $\mathbb{P}(\| C x_k\|\geq t)$ as a \textit{violation probability}.

The aim of this work is to design a controller that minimises the cost function
\begin{equation}
\Ex{\sum_{k=0}^\infty \norm{x_k-{x}^r}^2_Q+ \norm{u_k-u^r }^2_R} \label{eqn:cost:original form of the cost function}
\end{equation}
and ensures a quadratic stability condition on the closed loop system while the constraint \eqref{eqn:constraint:probability constraint of the original form} is satisfied.
The weighting matrices in~\eqref{eqn:cost:original form of the cost function} are assumed to satisfy $Q\succcurlyeq0$ and $R\succ0$, and we assume knowledge of 
reference targets ${x}^r$ and $u^r$ for the state and the control input that satisfy the steady state conditions
\begin{equation}
  \left(I-A\right)x^r=Bu^r , \quad \| C x^r \| <t .
  \label{eqn:reference value: model for reference state and control input}
\end{equation}

\begin{assumption}\label{assumption: controllability and observability}
$(A,B)$ is controllable and $(A,Q^{\frac{1}{2}})$ is observable.
\end{assumption}

\subsection{Finite horizon formulation}
The problem stated above employs an infinite horizon and is subject to a constraint defined on infinite horizon. If the infinite sequence of control inputs $\{u_k\}_{k=0}^{\infty}$ were considered to be decision variables, then clearly the optimisation problem would be infinite dimensional and thus in principle computationally intractable \cite{kouvaritakis2015mpcbook}. However, the use of an infinite horizon can impart desirable properties, notably stability \cite{Scokaert97,Mayne2000auto}. It is therefore beneficial to design an MPC law using a cost function and constraints that are defined on a finite horizon in such a way that they are equivalent to the infinite horizon cost and constraints of the original problem. The finite horizon cost function for a prediction horizon of $N$ steps is given by 
 \begin{equation}
\Ex{\sum_{i=0}^{N-1}\norm{x_{i|k}-{x}^r}^2_Q+ \norm{u_{i|k}-u^r }^2_R+F(x_{N|k})}
  \label{eqn:cost function:reformulated cost function using dual mode prediction}
\end{equation}
where $\Ex{F(x_{N|k})}$ is the terminal cost and $F(x)\geq 0$ for all $x$.
The constraint \eqref{eqn:constraint:probability constraint of the original form} is likewise truncated to a finite horizon:
\begin{equation}
  \sum_{i=0}^{N-1} \gamma^i \mathbb{P}\left( \norm{C x_{i|k}} \geq t \right)+ f(x_{N|k}) \leq \varepsilon_k .
  \label{eqn:constraint:probability constraint of the truncated form}
\end{equation}
Here $f(x_{N|k})$ is a terminal term 
chosen 
(as will be specified in \eqref{eq:term_constr} and Lemma \ref{lemma: expression for terminal term in constraint}) 
to approximate the infinite sum in \eqref{eqn:constraint:probability constraint of the original form} so that $\sum_{i=N}^{\infty}\gamma^i \mathbb{P} ( \|C x_{i|k}\| \geq t ) \leq f(x_{N|k})$, and $\varepsilon_k$ is a bound on the lhs of
\eqref{eqn:constraint:probability constraint of the truncated form}
that is achievable at time $k$. Although $\varepsilon_k$ may increase or decrease over time since it is conditioned on the system state at time $k$, we show in Section~\ref{section: epsilon sequence} that \eqref{eqn:constraint:probability constraint of the original form} is satisfied if $\varepsilon_0\leq e$ \mbox{and $\varepsilon_k$ is defined as described in Section~\ref{section:recursive feasibility}}.

Even with the cost and constraints defined as in \eqref{eqn:cost function:reformulated cost function using dual mode prediction}-\eqref{eqn:constraint:probability constraint of the truncated form} on a finite horizon, the probability distribution of states may be unknown at each time step and the finite horizon version of the problem is therefore still intractable in general. Even if the probability distribution of $\omega_k$ is known explicitly, computing~\eqref{eqn:cost function:reformulated cost function using dual mode prediction} and~\eqref{eqn:constraint:probability constraint of the truncated form} requires the solution of a set of multivariate convolution integrals, which in principle is still difficult to manage \cite{KOUVARITAKIS2010auto}.  

\subsection{Constraint handling and open loop optimisation}

This section considers how to approximate the finite horizon constraint \eqref{eqn:constraint:probability constraint of the truncated form} using the two-sided Chebyshev inequality~\cite[Section V.7]{feller71} and gives the explicit form of the MPC cost function. The cost and constraints are then combined to construct the MPC optimisation problem that is repeatedly solved online.
%
%
We define the sequence of control inputs predicted at time $k$ as
\begin{align}
&u_{i|k}=K( x_{i|k}-\bar{x}_{i|k})+m_{i|k},  & i&=0,\ldots,N-1 \label{eqn:control input:fisrt mode} \\
&u_{N+i|k}=K( x_{N+i|k}-{x}^r )+u^r, & i&=0,1,\ldots \label{eqn:control input:second mode} 
\end{align}
where $m_{i|k}$ is the $i$-step-ahead prediction of the nominal control input given information at time $k$, that is, $\Ex[k]{u_{i|k}}=m_{i|k}$, and $\bar{x}_{i|k}$ is the $i$-step-ahead prediction of the nominal state given information at time $k$, that is, $\Ex[k]{x_{i|k}}=\bar{x}_{i|k}$.

\begin{assumption}
$\Phi:=A+BK$ is strictly stable.
\end{assumption}

Given the predicted control law \eqref{eqn:control input:fisrt mode}-\eqref{eqn:control input:second mode}, the first two moments of the predicted state and control input sequences can be computed. Thus, the predicted nominal state trajectory is given by $\bar{x}_{0|k} = x_k$ and
\begin{align}
  & \bar{x}_{i|k}= A^i\bar{x}_{0|k}+\sum_{j=0}^{i-1} A^{i-1-j}Bm_{j|k} , & i&=1,\ldots,N
                                                                              \label{eqn:nominal state:dynamics describing the evolution of nominal state}\\
& \bar{x}_{N+i|k}= \Phi^i\left(\bar{x}_{N|k}-x^r\right)+x^r , &
                                                                i&=1,2,\ldots
                                                                   \label{eqn:nominal state:dynamics describing the evolution of nominal state after time N}
\end{align}
whereas the covariance matrix, $X_{i|k}$, of the $i$-step-ahead predicted state
is given by $X_{0|k}= 0$ and
\begin{equation}
X_{i|k}=\sum_{j=0}^{i-1} \Phi^jW\bigl(\Phi^j\bigr)^T,\quad i=1,2,\ldots . \label{eqn:covairance matrix of state:expression of covariance matrix in term of sums involving W}
\end{equation}
Clearly $X_{i|k}$ is independent of $k$, and in the following development we simplify notation by letting $\hat{X}_i := X_{i|k}$.

In this paper, we use Chebyshev's inequality to handle probabilistic constraints. The advantages of this approach are that it can cope with arbitrary or unknown disturbance probability distributions (the only information required being the first two moments of the predicted state trajectory), and furthermore it results in quadratic inequalities that are straightforward to implement. Approximating~\eqref{eqn:constraint:probability constraint of the truncated form} by direct application of the two-sided Chebyshev inequality \cite{Schildbach2015auto}, we obtain
\begin{align}
\frac{\tr ( C^T C\hat{X}_{i} )+\norm{ C\bar{x}_{i|k}}^2}{t^2} &\leq \beta_{i|k}, & i&=0,\ldots,N-1
  \label{eqn:constraint:new form,generated by chebyshev two-sided inequality}\\
\sum^{N-1}_{i=0}\gamma^i \beta_{i|k}+f(\bar{x}_{N|k}) &\leq \varepsilon_k,\label{eqn:constraint:new form,change infinite to finite horizon}
\end{align}
where $\{\beta_{i|k}\}_{i=0}^{N-1}$ is a sequence of non-negative scalars.
%
The terminal term $f(\bar{x}_{N|k})$ in~\eqref{eqn:constraint:new form,change infinite to finite horizon}  is chosen so that 
\begin{align}
f(\bar{x}_{N|k}) &= \frac{\tr(C^T C \widetilde{S})}{t^2}+
\frac{\gamma^N}{t^2} \biggl[\norm{\bar{x}_{N|k}-x^r}^2_{\widetilde{P}} 
+ \frac{\norm{x^r}^2_{C^TC}}{(1-\gamma)}\biggr] \nonumber \\
&+\frac{2\gamma^N (x^r)^TC^TC(I-\gamma \Phi)^{-1}(\bar{x}_{N|k}-x^r)}{t^2} 
\label{eq:term_constr}
\end{align}
where $\widetilde{S}\succ 0$, $\widetilde{P}\succ 0$, and $I-\gamma \Phi$ is invertible since $\gamma\Phi$ is strictly stable. The design of $\widetilde{S}, \widetilde {P}$ is discussed in Section \ref{section: epsilon sequence}.

In terms of the predicted nominal state trajectory in \eqref{eqn:nominal state:dynamics describing the evolution of nominal state}-\eqref{eqn:nominal state:dynamics describing the evolution of nominal state after time N}, the predicted cost is defined
\begin{multline}
  J(\bar{x}_{0|k},\{m_{i|k}\}_{i=0}^{N-1}, \varepsilon_k) := \norm{\bar{x}_{N|k}-x^r}^2_{P} \\
  + \sum_{i=0}^{N-1}\left(\norm{\bar{x}_{i|k}-{x}^r}^2_Q+ \norm{m_{i|k}-u^r }^2_R\right) 
\label{eqn:cost function: explicit cost function, final form}
\end{multline}
whenever a sequence $\{\beta_{i|k}\}_{i=0}^{N-1}$ exists satisfying \eqref{eqn:constraint:new form,generated by chebyshev two-sided inequality}-\eqref{eqn:constraint:new form,change infinite to finite horizon} for the  given $\bar{x}_{0|k}$, $\{m_{i|k}\}_{i=0}^{N-1}$ and $\varepsilon_k$. On the other hand, if $\bar{x}_{0|k}$, $\{m_{i|k}\}_{i=0}^{N-1}$ and $\varepsilon_k$ are such that constraints~\eqref{eqn:constraint:new form,generated by chebyshev two-sided inequality}-\eqref{eqn:constraint:new form,change infinite to finite horizon}  are infeasible, we set $ J(\bar{x}_{0|k},\{m_{i|k}\}_{i=0}^{N-1}, \varepsilon_k) := \infty$.
Note that $\| \bar{x}_{N|k}-x^r\|^2_{P}$ in \eqref{eqn:cost function: explicit cost function, final form} represents the terminal cost, and that $P \in \mathbb{S}^{n_{x}}_{++}$. The choice of $P$ is discussed in Section~\ref{section: smpc algorithm}.

To summarise, the MPC optimisation solved at time $k$ is
\begin{equation}\label{eq:mpc_optimisation}
J^\ast (x_k,\varepsilon_k) := \min_{\{m_{i|k}\}_{i=0}^{N-1}} J(x_k,\{m_{i|k}\}_{i=0}^{N-1}, \varepsilon_k) ,
\end{equation}
%
and its solution for any feasible $x_{k}$ and $\varepsilon_k$ is denoted
\begin{equation}
\bigl\{m^\ast_{i|k}(x_k,\varepsilon_k)\bigr\}_{i=0}^{N-1}:= \underset{\{m_{i|k}\}_{i=0}^{N-1}}{\arg\min} J\bigl(x_k,\{m_{i|k}\}_{i=0}^{N-1},\varepsilon_k\bigr). \label{eqn:minimiser of the open loop optimisation}
\end{equation}
For simplicity we write this solution as $\{m^\ast_{i|k}\}_{i=0}^{N-1}$,
with the understanding that this sequence depends on $x_k$ and $\varepsilon_k$. The corresponding nominal predicted state trajectory is given by
\begin{align}
  &\bar{x}^*_{i|k} = A^ix_{k}+\sum_{j=0}^{i-1} A^{i-1-j}Bm^*_{j|k} ,
  &
    i &= 1,\ldots,N
        \label{eqn:open loop system: optimal state trajectory, first mode}\\
  &\bar{x}^*_{N+i|k} = \Phi^i(\bar{x}^*_{N|k}-x^r)+x^r ,
  &
    i &= 1,2,\ldots .
        \label{eqn:open loop system: optimal state trajectory, second mode}
\end{align}
The MPC law at time $k$ is defined by
\begin{equation}
u_k := m^\ast_{0|k} ,  \label{eqn:closed loop control law}
\end{equation}
%
and the closed loop system dynamics are given by
\begin{equation}
  x_{k+1}=Ax_k+Bm^\ast_{0|k} (x_k, \varepsilon_k)+\omega_k ,
  \label{eqn:closed loop system dynamics, using the first element of optimal input sequence}
\end{equation}
where $\omega_k$ is the disturbance realisation at time $k$. 

In the remainder of this paper we discuss how to choose $\varepsilon_k$, $K$, $P$, $\widetilde{P}$ and $\widetilde{S}$ so as to guarantee quadratic stability and satisfaction of the constraint \eqref{eqn:constraint:probability constraint of the original form} under the MPC law \eqref{eqn:closed loop control law}.

\section{RECURSIVE FEASIBILITY} \label{section:recursive feasibility}

Recursively feasible MPC strategies have the property that the MPC optimisation problem is guaranteed to be feasible at every time-step if it is initially feasible.
This property can be ensured by imposing a terminal constraint that requires the predicted  system state to lie in a particular set at the end of the prediction horizon \cite{kouvaritakis2015mpcbook}. For a deterministic MPC problem, if an optimal solution can be found at current time, then the \textit{tail} sequence, namely the optimal control sequence shifted by one time-step, will be a feasible suboptimal solution at the next time instant if the terminal constraint is defined in terms of a suitable invariant set for the predicted system state \cite{Farina13,Lorenzen15}. For a robust MPC problem with bounded additive disturbances, recursive feasibility can likewise be guaranteed
under either open or closed loop optimisation strategies
by imposing a terminal constraint set that is robustly invariant. However, this approach is not generally applicable to systems with unbounded additive disturbances, and in general it is not possible to ensure recursive feasibility in this context while guaranteeing constraint satisfaction at every time instant.

In this section we propose a method for guaranteeing recursive feasibility of the MPC optimisation that does not rely on terminal constraints.
Instead recursive feasibility is ensured, despite the presence of unbounded disturbances, by allowing the constraint on the discounted sum of probabilities to be time-varying. For all time-steps $k > 0$, the approach uses the optimal sequence computed at time $k-1$ to determine a value of $\varepsilon_{k}$ that is necessarily feasible at time $k$. Using this approach it is possible to choose $\varepsilon_0$ so that the original constraint \eqref{eqn:constraint:probability constraint of the original form} is satisfied, as we discuss in Section~\ref{section: epsilon sequence}.

We use the notation $\mathscr{S}(\{m^\ast_{i|k}\}_{i=0}^{N-1})$ to denote a nominal control sequence derived from a time-shifted version of $\{m^\ast_{i|k}\}_{i=0}^{N-1}$, defined by
\begin{IEEEeqnarray}{C}
\mathscr{S}\left(\{m^\ast_{i|k}\}_{i=0}^{N-1}\right):=\{m^\ast_{i+1|k}+K\Phi^i\omega_k\}_{i=0}^{N-1},\label{eqn: shifted version of the solution to the optimisation of current time}
\end{IEEEeqnarray} 
with $m^\ast_{N|k}:={K} ( \bar{x}^\ast_{N|k}-{x}^r )+u^r$.
Note that the disturbance realisation $\omega_k$ can be computed given the measured state $x_{k+1}$ and hence the sequence $\mathscr{S}(\{m^\ast_{i|k}\}_{i=0}^{N-1})$ is available to the controller at time $k+1$.

\begin{lemma} \label{theorem:theorem for recursive feasibility}
  The MPC optimisation~\eqref{eq:mpc_optimisation} is recursively feasible if $\varepsilon_k$ is defined at each time $k=1,2,\ldots$ as
\begin{equation}
  \varepsilon_{k}:=\min \Big\{\varepsilon \mid J\Bigl(x_{k},\mathscr{S}\left(\{m^\ast_{i|k-1}\}_{i=0}^{N-1}\right),\varepsilon\Bigr)<\infty\Big\} . \label{eqn: updating epsilon: optimisation by solving which gives new epsilon}
\end{equation}
\end{lemma}


\begin{proof}
  The definition of the MPC predicted cost implies that, for any given sequence $\{m_{i|k}\}_{i=0}^{N-1}$, there necessarily exists a value of $\varepsilon$ such that $J(x_k,\{m_{i|k}\}_{i=0}^{N-1},\varepsilon)$ is finite.
  Moreover $\mathscr{S}(\{m^\ast_{i|k-1}\}_{i=0}^{N-1})$ is (with probability 1) well-defined at time $k$ if the MPC optimisation is feasible at time $k-1$. It follows that the minimum value of $\varepsilon$ defining $\varepsilon_k$ in \eqref{eqn: updating epsilon: optimisation by solving which gives new epsilon} exists if the MPC optimisation is feasible at time $k-1$, and this establishes recursive feasibility
\end{proof}

The sequence $\mathscr{S}(\{m^\ast_{i|k}\}_{i=0}^{N-1})$ can be regarded as the tail of the minimiser~\eqref{eqn:minimiser of the open loop optimisation} with adjustments. With equations \eqref{eqn:nominal state:dynamics describing the evolution of nominal state} and \eqref{eqn:nominal state:dynamics describing the evolution of nominal state after time N}, the minimisation \eqref{eqn: updating epsilon: optimisation by solving which gives new epsilon} can be solved to give an explicit expression for $\varepsilon_{k}$ for all $k > 0$ as
\begin{align}
\varepsilon_{k} &=\sum_{i=0}^{N-1}\gamma^i\frac{\tr\bigl(C^TC\hat{X}_{i}\bigr)+\bigl\| C \bigl(\bar{x}_{i+1|k-1}^\ast+\Phi^i\omega_{k-1}\bigr)\bigr\|^2}{t^2} \nonumber \\
&\quad +f\bigl(\bar{x}^\ast_{N+1|k-1}+\Phi^N\omega_{k-1}\bigr).\label{eqn:epsilon of next time instant:expression of new epsilon using results obtained at current time}
\end{align}


Essentially, the optimisation problem to be solved at each time step is feasible because the parameter $\varepsilon_k$ is updated via~\eqref{eqn:epsilon of next time instant:expression of new epsilon using results obtained at current time} using knowledge of the disturbance $w_{k-1}$ obtained from the measurement of the current state $x_k$. In this respect the approach is similar to constraint-tightening methods that have previously been applied in the context of stochastic MPC (e.g.~\cite{korda14,Lorenzen17ieeetr,fleming17}) in order to ensure recursive feasibility and constraint satisfaction in closed loop operation. However, each of these methods requires that the disturbances affecting the controlled system are bounded, and they become more conservative as the degree of conservativeness of the assumed disturbance bounds increases. The approach proposed here avoids this requirement and instead ensures closed loop constraint satisfaction using the analysis of Section~\ref{section: epsilon sequence}.

The key to this method lies in the definition of the sequence $\mathscr{S}(\{m^\ast_{i|k}\}_{i=0}^{N-1})$. If this sequence were optimised simultaneously with $\varepsilon_k$, rather than defined by the suboptimal control sequence~\eqref{eqn: shifted version of the solution to the optimisation of current time}, then
it would be possible to reduce the MPC cost~\eqref{eq:mpc_optimisation}. However this would require more computational effort than is needed to evaluate~\eqref{eqn:epsilon of next time instant:expression of new epsilon using results obtained at current time}.
For deterministic MPC problems it can be shown that the cost of using the tail sequence is no greater than the optimal cost at the current time with an appropriate terminal weighting matrix \cite{Mayne2000auto}, but this property cannot generally be ensured in the presence of unbounded disturbances. In fact the optimal cost defined by \eqref{eq:mpc_optimisation} is not necessarily monotonically non-increasing if $\varepsilon_k$ is defined by \eqref{eqn:epsilon of next time instant:expression of new epsilon using results obtained at current time}, but the proposed approach based on the adjusted tail sequence~\eqref{eqn: shifted version of the solution to the optimisation of current time}
ensures a quadratic closed loop stability bound, as we discuss next.


\section{SMPC ALGORITHM} \label{section: smpc algorithm}

This section analyses the stability of the MPC law and shows that the closed loop system satisfies a quadratic stability condition.
We first state the MPC algorithm based on the optimisation defined in~\eqref{eq:mpc_optimisation}.
\vspace{2mm}
\begin{algorithm}\label{algorithm:SPMC algorithm}
  At each time-step $k=0,1,\ldots$:
  \begin{enumerate}[(i).] 
  \item
    Measure $x_k$, and if $k>0$, compute $\varepsilon_k$ using \eqref{eqn:epsilon of next time instant:expression of new epsilon using results obtained at current time}.
  \item Solve the quadratically constrained quadratic programming (QCQP) problem:
    \[
      \begin{aligned}
        \minimise_{\{m_{i|k},\, \beta_{i|k} \}_{i=0}^{N-1}} \ & 
        \sum_{i=0}^{N-1}\Bigl(\bigl\|\bar{x}_{i|k}-{x}^r\bigr\|^2_Q+ \bigl\| m_{i|k}-u^r \bigr\|^2_R\Bigr)  \\
        & + \bigl\|\bar{x}_{N|k}-x^r\bigr\|^2_{P}  \\
      \end{aligned}
    \]
    subject to \eqref{eqn:constraint:new form,generated by chebyshev two-sided inequality}, \eqref{eqn:constraint:new form,change infinite to finite horizon}, and \eqref{eqn:nominal state:dynamics describing the evolution of nominal state} with $\bar{x}_{0|k} = x_k$.
  \item Apply the control law $u_k = m_{0|k}^\ast$.
  \end{enumerate}
\end{algorithm}
\vspace{2mm}


Although the MPC optimisation in step (ii) involves a quadratic constraint as well as linear constraints, it can be solved efficiently, for example using a second-order conic program (SOCP) solver, since the objective and the quadratic constraint are both convex.

\begin{theorem} \label{theorem:stability result}
Given initial feasibility at $k=0$, the minimisation in step (ii) of Algorithm~\ref{algorithm:SPMC algorithm} is feasible for $k=1,2,\ldots$ and the closed loop system satisfies the quadratic stability condition
\begin{equation}
\lim_{T\to\infty}\frac{1}{T}\sum_{k=0}^{T-1}\Ex{\norm{{x}_{k}-{x}^r}^2_Q+\norm{u_{k}-u^r}^2_R} 
\leq \tr (W P) \label{eqn:quadratic stability result}
\end{equation}
provided $K$ in \eqref{eqn:control input:fisrt mode}-\eqref{eqn:control input:second mode} and $P$ in \eqref{eqn:cost function: explicit cost function, final form} are chosen so that 
\begin{equation}
P = \Phi^TP \Phi + K^TR K +Q. \label{eqn:discrete time Lyapunov equation for compute penalty matrix in terminal cost}
\end{equation}
\end{theorem}

\begin{proof}
  From Lemma \ref{theorem:theorem for recursive feasibility}, the sequence $\mathscr{S}\bigl(\{m^\ast_{i|k}\}_{i=0}^{N-1}\bigr)$ provides a feasible suboptimal solution at time $k+1$. Hence by optimality we necessarily have
  \[
    J^\ast(x_{k+1},\varepsilon_{k+1}) \leq
    J\bigl(x_{k+1}, \mathscr{S} \bigl(\{m^\ast_{i|k}\}_{i=0}^{N-1}\bigr), \varepsilon_{k+1}\bigr),
  \]
and since this inequality holds for every realisation of $\omega_{k}$, by taking expectations conditioned on the state $x_k$ we obtain
\begin{equation}
\mathbb{E}_k\bigl[ J^\ast (x_{k+1},\varepsilon_{k+1})\bigr] \! \! \leq \! \mathbb{E}_k\bigl[ J \bigl(x_{k+1}, \mathscr{S} \bigl(\{m^\ast_{i|k}\}_{i=0}^{N-1}\bigr),\varepsilon_{k+1}\bigr) \bigr] \!.
\label{eqn:expected_cost_diff}
\end{equation}
Evaluating $\bar{x}_{i|k+1}$ by setting $\bar{x}_{0|k+1} = x_{k+1}$ and $m_{i|k+1} = m^\ast_{i+1|k} + K \Phi^i \omega_k$ in
\eqref{eqn:nominal state:dynamics describing the evolution of nominal state}-\eqref{eqn:nominal state:dynamics describing the evolution of nominal state after time N} gives the feasible sequence
\[
  \bar{x}_{i|k+1} = \bar{x}^\ast_{i+1|k} + \Phi^i \omega_k , \qquad i = 0,\ldots,N ,
\]
and from \eqref{eqn:discrete time Lyapunov equation for compute penalty matrix in terminal cost} and \eqref{eqn:expected_cost_diff} it follows that
\begin{multline}
  \Ex[k]{ J^\ast (x_{k+1},\varepsilon_{k+1})} \leq J^\ast (x_k,\varepsilon_k) - \| x_k - x^r \|_Q^2 \\- \| u_k - u^r \|_R^2 + \tr(W P) .
  \label{eqn:cost comparison in mean, one optimal and one feasible at time k+1}
\end{multline}
Summing both sides of this inequality over $k\geq 0$ after taking expectations given information available at time ${k=0}$, and making use of the property that $\Ex[0]{\Ex[k]{J^\ast (x_{k+1},\varepsilon_{k+1})}} = \Ex[0]{J^\ast (x_{k+1},\varepsilon_{k+1})}$, gives \eqref{eqn:quadratic stability result}. 
\end{proof}

Stability is the overriding requirement and in most recent MPC literature the cost function is chosen so as to provide a Lyapunov function suitable for analysing closed loop stability~\cite{Mayne2000auto}. Theorem~\ref{theorem:stability result} is proved via cost comparison, and, given the quadratic form of the cost function, this analysis results in the quadratic stability condition~\eqref{eqn:quadratic stability result}. 
Similar asymptotic bounds on the time average of a quadratic expected stage cost are obtained in \cite{KOUVARITAKIS2010auto,Cannon09ieetr}.
However, in the current context, Theorem~\ref{theorem:stability result} demonstrates that an MPC algorithm can ensure closed loop stability without imposing terminal constraints derived from an invariant set.
 
\begin{lemma} \label{theorem:convergence of closed loop controller to unconstrained LQ optimal}
  If $K$ in \eqref{eqn:control input:fisrt mode}-\eqref{eqn:control input:second mode} is the unconstrained LQ-optimal feedback gain, $K_{LQ}$, for the system \eqref{eqn:system dynamics:original state space model} with cost~\eqref{eqn:cost:original form of the cost function}, then for the closed loop system under the control strategy of Algorithm~\ref{algorithm:SPMC algorithm},
the control law $u_k = m_{0|k}^\ast$ converges  as $k\to\infty$ to the unconstrained optimal feedback law $u_k = K_{LQ}x_k$.
\end{lemma}
\begin{proof}
  Consider a system with the same model parameters $A,B,W$ as~\eqref{eqn:system dynamics:original state space model}, and a stabilizing linear feedback law with gain $K$.
  Denoting the states and control inputs of this system respectively as $\hat{x}_k$ and $\hat{u}_k = K\hat{x}_k$, we have
  \begin{equation}  \lim_{k\to\infty}\mathbb{E}
    \Bigl[ \norm{\hat{x}_{k}-x^r}^2_Q+\norm{\hat{u}_{k}-u^r}^2_R\Bigr]=\tr(WP)
  \label{eqn: another closed loop result: value of the expectation of cost at time infinity using LQ-optimal}
\end{equation}
where $P$ is the solution of \eqref{eqn:discrete time Lyapunov equation for compute penalty matrix in terminal cost}.  
%
However the certainty equivalence theorem \cite{Chow75dynamic} implies that $\tr(WP)$ is minimized with $K=K_{LQ}$. Therefore \eqref{eqn:quadratic stability result} implies 
\begin{align*}
\lim_{T\to\infty}\frac{1}{T}\sum_{k=0}^{T-1} & \mathbb{E}\Bigl[\norm{{x}_{k}-{x}^r}^2_Q+\norm{u_{k}-u^r}^2_R\Bigr] \\
&=\lim_{T\to\infty}\mathbb{E}\Bigl[\norm{\hat{x}_{T}-x^r}^2_Q+\norm{\hat{u}_{T}-u^r}^2_R\Bigr],
\end{align*}
so that $u_k \to K_{LQ} x_k$ as $k\to\infty$ under Assumption \ref{assumption: controllability and observability}.
\end{proof}

The convergence result in Lemma~\ref{theorem:convergence of closed loop controller to unconstrained LQ optimal} is to be expected because of the discounted constraint~\eqref{eqn:constraint:probability constraint of the original form}.
Since $\gamma^k \rightarrow 0$ as $k \rightarrow \infty$, the probabilistic constraint places greater emphasis on near-future predicted states and ignores asymptotic behaviour.
Under this condition the unconstrained LQ-optimal feedback control law is asymptotically optimal for~\eqref{eqn:cost:original form of the cost function}.

\section{THE BEHAVIOUR OF THE SEQUENCE $\{\varepsilon_k\}_{k=0}^{\infty}$ AND CONSTRAINT SATISFACTION} \label{section: epsilon sequence}

This section considers the properties of the sequence $\{\varepsilon_k\}_{k=0}^{\infty}$ in closed loop operation under Algorithm~\ref{algorithm:SPMC algorithm}.
We first give expressions for the parameters $\widetilde{S}$ and $\widetilde{P}$ in the definition~\eqref{eq:term_constr} of the terminal term $f(\bar{x}_{N|k})$.
Then, using the explicit expression for $\varepsilon_k$ in~\eqref{eqn:epsilon of next time instant:expression of new epsilon using results obtained at current time}, 
we derive a recurrence equation relating the expected value of $\varepsilon_{k+1}$ to $x_k$ and $\varepsilon_k$. This allows an upper bound to be determined for the sum of discounted violation probabilities on the left hand side of~\eqref{eqn:constraint:probability constraint of the original form}. With this bound we can show that the closed loop system under the control law of Algorithm~\ref{algorithm:SPMC algorithm} satisfies the chance constraint~\eqref{eqn:constraint:probability constraint of the original form} if $\varepsilon_k$ is initialised with $\varepsilon_0=e$.

\begin{lemma} \label{lemma: expression for terminal term in constraint}
  Let $\widetilde{S}$ and $\widetilde{P}$ be the solutions of
  \begin{align}
    &\widetilde{P} = \gamma\Phi^T \widetilde{P} \Phi + C^TC \label{eq:Plyap}\\
    &\widetilde{S} = \gamma\Phi \widetilde{S} \Phi^T + \frac{\gamma^{N+1}}{1-\gamma} W + \gamma^N \hat{X}_N . \label{eq:Slyap}
  \end{align}
  Then $f(\bar{x}_{N|k})$ defined in~\eqref{eq:term_constr} satisfies
  \begin{equation}\label{eq:fdef}
    f(\bar{x}_{N|k})  = \sum_{i=N}^\infty \gamma^i \frac{\tr \bigl( C^T C\hat{X}_{i}\bigr)+\norm{ C\bar{x}_{i|k}}^2}{t^2}
  \end{equation}
  where $\bar{x}_{i|k}$ is given by~\eqref{eqn:nominal state:dynamics describing the evolution of nominal state after time N} for all $i\geq N$.
\end{lemma}

\begin{proof}
  Writing $\|C\bar{x}_{i|k}\|^2 = \|C(\bar{x}_{i|k} - x^r) + Cx^r\|^2$ and using~\eqref{eqn:nominal state:dynamics describing the evolution of nominal state after time N}, we obtain
  \begin{align*}
    \norm{C\bar{x}_{i|k}}^2 & = \norm{C\Phi^{i-N}(\bar{x}_{N|k} - x^r)}^2 \\
    &\quad+ 2 (x^r)^T C^T C \Phi^{i-N}(\bar{x}_{N|k}-x^r) + \norm{Cx^r}^2 
  \end{align*}
for all $i\geq N$,
  and since $\smash{\widetilde{P}}$ satisfies \eqref{eq:Plyap}, we have
  \begin{IEEEeqnarray}{C}\label{eq:fproof1}
    \sum_{i=N}^\infty \frac{\gamma^i}{t^2} \norm{C\bar{x}_{i|k}}^2 \! = \! 
    \frac{\gamma^N}{t^2} \norm{\bar{x}_{N|k}-x^r}^2_{\widetilde{P}} 
    +\frac{\gamma^N}{(1-\gamma)} \frac{\norm{x^r}^2_{C^TC}}{t^2}
    \nonumber \\
    +\frac{2\gamma^N (x^r)^TC^TC(I-\gamma \Phi)^{-1}(\bar{x}_{N|k}-x^r)}{t^2}.
  \end{IEEEeqnarray}
  Furthermore, if $\widetilde{S} = \sum_{i=N}^\infty \gamma^i \hat{X}_i$, then $\widetilde{S}$ is the solution of the Lyapunov equation~\eqref{eq:Slyap}  since~\eqref{eqn:covairance matrix of state:expression of covariance matrix in term of sums involving W} implies
  \begin{align*}
    \gamma \Phi \widetilde{S} \Phi^T &= \sum_{i=N}^\infty \gamma^{i+1} \Phi \hat{X}_i \Phi^T = \sum_{i=N}^\infty \gamma^{i+1}(\hat{X}_{i+1} - W) \\
                                     &= \widetilde{S} - \gamma^N\hat{X}_N - \frac{\gamma^{N+1}}{1-\gamma} W ,
  \end{align*}
  and it follows that
  \begin{equation}\label{eq:fproof2}
    \sum_{i=N}^\infty \frac{\gamma^i}{t^2} \tr\bigl(C^TC\hat{X}_i\bigr) = \frac{\tr\bigl(C^TC\widetilde{S}\bigr)}{t^2} .
\end{equation}
Combining~\eqref{eq:fproof1} and~\eqref{eq:fproof2}, it is clear that~\eqref{eq:fdef} is equivalent to~\eqref{eq:term_constr} if $\widetilde{P}$ and $\widetilde{S}$ are defined by~\eqref{eq:Plyap} and~\eqref{eq:Slyap}.
\end{proof}

The following result gives the relationship between  $\varepsilon_k$ and the expected value of $\varepsilon_{k+1}$ for the closed loop system.

\begin{theorem}\label{thm:epsilon}
  If $\varepsilon_k$ is defined by~\eqref{eqn:epsilon of next time instant:expression of new epsilon using results obtained at current time} at all times $k\geq1$, then   in closed loop operation under Algorithm~\ref{algorithm:SPMC algorithm} we have \begin{equation}\label{eq:epsilon}
    \gamma \Ex[k]{\varepsilon_{k+1}} \leq \varepsilon_k - \frac{\norm{Cx_k}^2}{t^2}
  \end{equation}
  for all $k\geq 0$.
\end{theorem}

\begin{proof}
  Evaluating $\varepsilon_{k+1}$ using~\eqref{eqn:epsilon of next time instant:expression of new epsilon using results obtained at current time} and~(\ref{eq:fdef}) gives
  \[
    \varepsilon_{k+1} = \sum_{i=0}^\infty \gamma^i \frac{\tr\bigl(C^TC \hat{X}_i\bigr) + \bigl\|C \bigl(\bar{x}_{i+1|k}^\ast + \Phi^i \omega_k\bigr)\bigr\|^2}{t^2} ,
  \]
  where $\bar{x}^\ast_{i|k}$ is given by~\eqref{eqn:open loop system: optimal state trajectory, first mode}-\eqref{eqn:open loop system: optimal state trajectory, second mode} and $\omega_k$ is the realisation of the disturbance at time $k$. Taking expectations conditioned on information available at time $k$, this implies
  \begin{multline*}
    \gamma \Ex[k]{\varepsilon_{k+1}} = \sum_{i=0}^\infty \gamma^{i+1} 
    \frac{\tr\bigl(C^TC \hat{X}_i\bigr) + \bigl\| C\bar{x}^\ast_{i+1|k}\bigr\|^2 }{t^2} \\
    + \sum_{i=0}^\infty
    \gamma^{i+1}\frac{\tr\bigl( C^T
      C \Phi^i W \left( \Phi^i \right)^T\bigr)}{t^2} ,
  \end{multline*}
  but feasibility of the sequence $\{m_{i|k}^\ast\}_{i=0}^{N-1}$ at time $k$ implies
  $\sum_{i=0}^\infty \frac{\gamma^{i} }{t^2}
    \bigl[\tr(C^TC \hat{X}_i) + \| C\bar{x}^\ast_{i|k} \|^2 \bigr]
    \leq \varepsilon_k$
  and therefore
  \begin{multline*}
    \gamma\Ex[k]{\varepsilon_{k+1}} \leq \varepsilon_k - \frac{\norm{Cx_k}^2}{t^2} \\
    + \sum_{i=0}^\infty \frac{\gamma^i}{t^2}
    \tr\bigl[ C^TC \bigl( \gamma \Phi^i W {\Phi^i}^T +  (\gamma - 1) \hat{X}_i \bigr)\bigr] .
  \end{multline*}
To complete the proof we note that the sum on the RHS of this inequality is zero since
  \[
    \gamma^{i+1} \Phi^i W {\Phi^i}^T +  (\gamma^{i+1} - \gamma^i) \hat{X}_i = \gamma^{i+1} \hat{X}_{i+1} - \gamma^i \hat{X}_i,
  \]
and because $\hat{X}_0=0$ and $\lim_{i\to\infty} \gamma^i \hat{X}_i = 0$.
\end{proof}

The main result of this section is given next.
\begin{theorem} \label{theorem:closedloop_cc}
  The closed loop system under Algorithm~\ref{algorithm:SPMC algorithm} satisfies the chance constraint \eqref{eqn:constraint:probability constraint of the original form} if $\varepsilon_0=e$. 
\end{theorem}
\begin{proof}
  Theorem~\ref{thm:epsilon} implies that the closed loop evolution of $\varepsilon_k$ satisfies
  \[
    \gamma^{i+1} \Ex[k]{\varepsilon_{k+i+1}} \leq \gamma^{i} \Ex[k]{\varepsilon_{k+i}} - \frac{\gamma^{i}}{t^2} \Ex[k]{\| C x_{k+i} \|^2}
  \]
  for all non-negative integers $k,i$. Summing both sides of this equation over $i\in\{0,1,\ldots\}$ gives
  \begin{equation} \label{eqn: Markov inequality bound}
  \varepsilon_k \geq \sum_{i=0}^\infty \gamma^i\frac{\Ex[k]{ \| C x_{k+i} \|^2}}{t^2}
    + \lim_{i\to\infty} \gamma^i \Ex[k]{\varepsilon_{k+i}} .
  \end{equation}
  But $\gamma^i \Ex[k]{\varepsilon_{k+i}}$ is necessarily non-negative for all $k,i\geq 0$, so by Chebyshev's inequality this implies   
  \begin{equation}\label{eq:closedloop_cc}
    \sum_{i=0}^\infty \gamma^i \mathbb{P}\bigl( \norm{C x_{k+i}} \geq t \bigr) \leq \varepsilon_k
  \end{equation}
  for all $k\geq 0$. An obvious consequence of the bound \eqref{eq:closedloop_cc} is that
the closed loop system will satisfy the chance constraint~\eqref{eqn:constraint:probability constraint of the original form}
if $\varepsilon_0$ is chosen to be equal to $e$.
\end{proof}
\vspace{2mm}

%
%

The presence of the factor $\gamma\in(0,1)$ on the LHS of~\eqref{eq:epsilon} implies that the expected value of $\varepsilon_k$ can increase as well as decrease along closed loop system trajectories. In fact, for values of $\gamma$ close to zero, a rapid initial growth in $\varepsilon_k$ is to be expected, which is in agreement with the interpretation that the constraint~\eqref{eqn:constraint:probability constraint of the original form} penalises violation probabilities more heavily at times closer to the initial time in this case. On the other hand, for values of $\gamma$ close to 1, $\varepsilon_k$ can be expected to decrease initially, implying a greater emphasis on the expected number of violations over some initial horizon.

\section{NUMERICAL EXAMPLE} \label{section:numerical example}
This section describes a numerical example illustrating the quadratic stability and  constraint satisfaction of the closed loop system \eqref{eqn:closed loop system dynamics, using the first element of optimal input sequence} under Algorithm~\ref{algorithm:SPMC algorithm}. Consider a system
with
\begin{IEEEeqnarray*}{C}
{\footnotesize 
A =
\begin{bmatrix}
1  & 2 \\
1.5 & 0.5
\end{bmatrix},
\quad B=
\begin{bmatrix}
1.2\\
1.5
\end{bmatrix}},
\end{IEEEeqnarray*}
and Gaussian disturbance $\omega_k\sim\mathcal{N}(0,W)$ with covariance matrix $W = 0.2 I_{2\times 2}$.
The constraint~\eqref{eqn:constraint:probability constraint of the original form} is defined by
$\gamma=0.9,\,t=1,\,e=3.5,\,C=
\begin{bmatrix}
0.6  & 
0.52
\end{bmatrix},
$
and the weighting matrices in the cost \eqref{eqn:cost:original form of the cost function} are given by
{\footnotesize \[ 
Q= C^TC =
\begin{bmatrix}
    0.3600&    0.3120\\
    0.3120&    0.2704
\end{bmatrix},
\quad R=1 .
\]}
Input and state references are $u^r=-0.6$, $x^r= (0.72, 0.36)$,
and the prediction horizon is chosen as $N=7$. The feedback gain is chosen as
$K=[{-0.92} \ {-0.85}]$ for the cost~\eqref{eqn:cost:original form of the cost function}, and matrices $P$, $\widetilde{P}$ and $\widetilde{S}$ are chosen to satisfy~\eqref{eqn:discrete time Lyapunov equation for compute penalty matrix in terminal cost}, \eqref{eq:Plyap} and \eqref{eq:Slyap}. The initial value for $\varepsilon_k$ is $\varepsilon_0=e=3.5$. 

Two sets of simulations (\textit{A} and \textit{B}) demonstrate the closed loop stability result in Theorem~\ref{theorem:stability result} and the constraint satisfaction result in Theorem~\ref{theorem:closedloop_cc}, respectively.

\textit{Simulation A}:\
To estimate empirically the average cost, 
$$
  \bar{J} :=
\lim_{T\to\infty}\frac{1}{T}\sum_{k=0}^{T-1}\mathbb{E}\Bigl[\norm{{x}_{k}-{x}^r}^2_Q+\norm{u_{k}-u^r}^2_R\Bigr],
$$
we consider the mean value of the stage cost over $100$ simulations. Each simulation has a randomly selected initial condition ($x_0\sim\mathcal{N}(0,I)$, with infeasible values discarded), and a length of $T = 500$ time-steps.
This gives the estimated average cost as $\bar{J} \approx 0.5036$, which is no greater than $\tr(WP)=0.5304$, and hence agrees with the bound~\eqref{eqn:quadratic stability result}. Moreover, the estimate of $\bar{J}$ decreases considerably more slowly as the simulation length $T$ increases.

\textit{Simulation B}:\ 
To test numerically whether the chance constraint \eqref{eqn:constraint:probability constraint of the original form} is satisfied, we estimate the discounted sum of violation probabilities on the LHS of \eqref{eqn:constraint:probability constraint of the original form},
$$
  V \!:= \sum_{k=0}^\infty \gamma^k \mathbb{P}(\|Cx_k \|\!\geq t ),
$$
by counting the number of violations at $k\in\{0,\ldots,T-1\}$, for $10^3$ simulations with $x_0=(-1.1130, 1.1156)^T$ and $T=100$. This gives $V\approx 0.8328$, which is less than $e = 3.5$ and hence satisfies the constraint \eqref{eqn:constraint:probability constraint of the original form}. For this example we have $\gamma^{100}\approx10^{-5}$, so increasing $T$ beyond $100$ time-steps has negligible effect on the estimate of $V$. Therefore the discrepancy between $e$ and the estimated value of $V$ can be attributed to the conservativeness of Chebyshev's inequality. In addition, if the unconstrained LQ-optimal feedback law $u_k=K_{LQ}(x_k-x^r)+u^r$ were employed, the value of the bound $\sum_{k=0}^\infty \gamma^k\Ex[k]{ \| C x_{k} \|^2}/t^2$ in \eqref{eqn: Markov inequality bound} would be $4.6998$, which exceeds $e$. 
Hence this control law may not satisfy~\eqref{eqn:constraint:probability constraint of the original form} and is worse than the MPC law \eqref{eqn:closed loop control law} in terms of this bound.

\section{CONCLUSIONS}\label{section:conclusion}

A stochastic MPC algorithm that imposes constraints on the sum of discounted future constraint violation probabilities can ensure recursive feasibility of the online optimisation and closed loop constraint satisfaction. Key features are the design of a constraint-tightening procedure and closed loop analysis of the tightening parameters. The MPC algorithm requires knowledge of the first and second moments of the disturbance, and is implemented as a convex QCQP problem. 
%






\bibliography{Paper_draft}

\begin{thebibliography}{10}

\bibitem{blanchini91}
F.~Blanchini, ``Constrained control for uncertain linear systems,'' {\em J.\
  Optim.\ Theory Appl.}, vol.~71, pp.~465--484, 1991.

\bibitem{kothare96}
M.~Kothare, V.~Balakrishnan, and M.~Morari, ``Robust constrained model
  predictive control using linear matrix inequalities,'' {\em Automatica},
  vol.~32, no.~10, pp.~1361--1379, 1996.

\bibitem{mayne05}
D.~Mayne, M.~Seron, and S.~Rakovi{\'c}, ``Robust model predictive control of
  constrained linear systems with bounded disturbances,'' {\em Automatica},
  vol.~41, no.~2, pp.~219--224, 2005.

\bibitem{schwarm99}
A.~Schwarm and M.~Nikolaou, ``Chance-constrained model predictive control,''
  {\em AIChE Journal}, vol.~45, no.~8, pp.~1743--1752, 1999.

\bibitem{KOUVARITAKIS2010auto}
B.~Kouvaritakis, M.~Cannon, S.~Rakovi{\'c}, and Q.~Cheng, ``Explicit use of
  probabilistic distributions in linear predictive control,'' {\em Automatica},
  vol.~46, no.~10, pp.~1719--1724, 2010.

\bibitem{kouvaritakis2015mpcbook}
B.~Kouvaritakis and M.~Cannon, {\em Model Predictive Control: Classical, Robust
  and Stochastic}.
\newblock Springer, 2015.

\bibitem{korda14}
M.~Korda, R.~Gondhalekar, F.~Oldewurtel, and C.~Jones, ``Stochastic {MPC}
  framework for controlling the average constraint violation,'' {\em IEEE
  Trans.\ Autom.\ Control}, vol.~59, pp.~1706--1721, 2014.

\bibitem{Lorenzen17ieeetr}
M.~Lorenzen, F.~Dabbene, R.~Tempo, and F.~Allg\"{o}wer, ``Constraint-tightening
  and stability in stochastic model predictive control,'' {\em IEEE Trans.\
  Autom.\ Control}, vol.~62, pp.~3165--3177, 2017.

\bibitem{fleming17}
J.~Fleming and M.~Cannon, ``Time-average constraints in stochastic {M}odel
  {P}redictive {C}ontrol,'' in {\em American Control Conference},
  pp.~5648--5653, 2017.

\bibitem{FRANKEL2016}
A.~Frankel, ``Discounted quotas,'' {\em Journal of Economic Theory}, vol.~166,
  pp.~396 -- 444, 2016.

\bibitem{kouvaritakis06}
B.~Kouvaritakis, M.~Cannon, and P.~Couchman, ``{MPC} as a tool for sustainable
  development integrated policy assessment,'' {\em IEEE Trans. Automatic
  Control}, vol.~51, no.~1, pp.~145--149, 2006.

\bibitem{kamgarpour17}
M.~Kamgarpour and T.~Summers, ``On infinite dimensional linear programming
  approach to stochastic control,'' in {\em Proceedings of the IFAC World
  Congress}, 2017.

\bibitem{cannon2011ieee}
M.~Cannon, B.~Kouvaritakis, S.~Rakovi{\'c}, and Q.~Cheng, ``Stochastic tubes in
  model predictive control with probabilistic constraints,'' {\em IEEE Trans.\
  Autom.\ Control}, vol.~56, no.~1, pp.~194--200, 2011.

\bibitem{Scokaert97}
P.~Scokaert, ``Infinite horizon generalized predictive control,'' {\em Int.\
  J.\ Control}, vol.~66, no.~1, pp.~161--175, 1997.

\bibitem{Mayne2000auto}
D.~Mayne, J.~Rawlings, C.~Rao, and P.~Scokaert, ``Constrained model predictive
  control: Stability and optimality,'' {\em Automatica}, vol.~36, no.~6,
  pp.~789--814, 2000.

\bibitem{feller71}
W.~Feller, {\em An Introduction to Probability Theory and Its Applications},
  vol.~2.
\newblock John Wiley \& Sons, 2nd~ed., 1971.

\bibitem{Schildbach2015auto}
G.~Schildbach, P.~Goulart, and M.~Morari, ``Linear controller design for chance
  constrained systems,'' {\em Automatica}, vol.~51, pp.~278--284, 2015.

\bibitem{Farina13}
M.~Farina, L.~Giulioni, L.~Magni, and R.~Scattolini, ``A probabilistic approach
  to model predictive control,'' in {\em IEEE Conference on Decision and
  Control}, pp.~7734--7739, 2013.

\bibitem{Lorenzen15}
M.~Lorenzen, F.~Allg\"{o}wer, F.~Dabbene, and R.~Tempo, ``An improved
  constraint-tightening approach for stochastic {MPC},'' in {\em American
  Control Conference}, pp.~944--949, 2015.

\bibitem{Cannon09ieetr}
M.~Cannon, B.~Kouvaritakis, and X.~Wu, ``Probabilistic constrained {MPC} for
  multiplicative and additive stochastic uncertainty,'' {\em IEEE Trans.\
  Autom.\ Control}, vol.~54, pp.~1626--1632, 2009.

\bibitem{Chow75dynamic}
G.~Chow and S.~Goldfeld, {\em Analysis and Control of Dynamic Economic
  Systems}.
\newblock Wiley, 1975.

\end{thebibliography}
\bibliographystyle{ieeetr}

\end{document}